\documentclass[11pt,letterpaper]{article}
\usepackage[utf8]{inputenc}

\usepackage[export]{adjustbox}
\usepackage{algorithm}
\usepackage{algorithmicx,algpseudocode}
\usepackage{amsmath,amssymb,amsfonts,amsthm}
\usepackage{array}
\usepackage{bibunits}
\usepackage{booktabs}
\usepackage{caption}
\usepackage{cite}
\usepackage{colortbl}
\usepackage{csvsimple}
\usepackage{etoolbox}
\usepackage[T1]{fontenc} 
\usepackage[margin=1in]{geometry}
\usepackage{graphicx}
\usepackage{glossaries,glossaries-extra}
\usepackage[hypertexnames=false]{hyperref}
\usepackage{cleveref}
\usepackage{ifthen}
\usepackage{import}
\usepackage{listings}
\usepackage{longtable}
\usepackage{makecell}
\usepackage{makecmds}
\usepackage{multirow}
\usepackage{natbib}
\usepackage[section]{placeins}
\usepackage{pythonhighlight}
\usepackage{rotating}
\usepackage[moderate]{savetrees}
\usepackage{setspace}
\usepackage{siunitx}
\usepackage{stringstrings}
\usepackage{subcaption}
\usepackage{subcaption}
\usepackage{tabularx}
\usepackage{textcomp}
\usepackage{vcell}
\usepackage{xcolor}
\usepackage{xfp}
\usepackage{xparse}
\usepackage{xstring}

\renewcommand{\algorithmicrequire}{\textbf{Input:}}

\makeatletter
\let\pragma@iinput=\@iinput
\def\@iinput#1{\xdef\@pragmafile{#1}\pragma@iinput{#1} }
\def\@pragmafile{default}
\def\pragmaonce{%
   \csname pragma@\@pragmafile\endcsname
   \global\expandafter\let \csname pragma@\@pragmafile\endcsname =  
}
\makeatother

\Crefname{ALC@unique}{Line}{Lines}

\ifdefined\mydraft
\mydraft
\fi

\defaultbibliography{bibl}
\defaultbibliographystyle{apalike}

\graphicspath{{./}{./hereditary-stratigraph-concept/tex}{./hstrat-evolutionary-inference}}

\newtheorem{theorem}{Theorem}

\makeglossaries

\newglossaryentry{differentia}
{
    name=differentia,
    description={randomly generated information that can be used to differentiate two strata at the same layer (with a probability of spurious collision)}
}

\newglossaryentry{stratum}
{
    name=stratum,
    description={differentia container associated with a particular historical stage}
}

\newglossaryentry{deposit}
{
    name=deposit,
    description={the act of extending the historical record with a new stratum}
}

\newglossaryentry{deposition}
{
    name=deposition,
    description={the stratum appended to the historical record}
}

\newglossaryentry{time point}
{
    name=time point,
    description={refers to the stage where a specific number of depositions have taken place}
}

\newglossaryentry{deposition time}
{
    name=deposition time,
    description={refers to the time point at which a stratum was deposited, zero indexed}
}

\newglossaryentry{genesis}
{
    name=genesis,
    description={the time point associated with the first stratum deposition}
}

\newglossaryentry{time}
{
    name=time,
    description={number of depositions elapsed between two time points}
}

\newglossaryentry{layer age}
{
    name=layer age,
    description={the number of depositions elapsed since a layer's deposition time}
}

\newglossaryentry{record depth}
{
    name=record depth,
    description={the number of depositions elapsed onto the historical record --- the number of layers within a historical record}
}

\newglossaryentry{layer}
{
    name=layer,
    description={position within a historical record associated with a single time point, which may or may be occupied}
}

\newglossaryentry{retained/pruned layer}
{
    name=retained/pruned layer,
    description={a layer within a historical record with present/removed strata, respectively}
}

\newglossaryentry{layer time point}
{
    name=layer time point,
    description={the deposition time associated with a layer}
}

\newglossaryentry{pruning}
{
    name=pruning,
    description={deletion of strata from a historical record (i.e., to reduce space occupied by the record) --- also used to refer to deletion of perfect tracking records for extinct lineages,}
}

\newglossaryentry{retention}
{
    name=retention,
    description={the act of carrying over a stratum into the next time point during the stratum deposition process}
}

\newglossaryentry{gap}
{
    name=gap,
    description={layers associated with contiguous time points that have been pruned ---- introduces inference uncertainty when comparing two columns}
}

\newglossaryentry{gap width}
{
    name=gap width,
    description={the number of contiguous time points that have been pruned --- gap with increases inference uncertainty}
}

\newglossaryentry{sparse/dense retention}
{
    name=sparse/dense retention,
    description={refers to relatively wide or relatively tight gap width, respectively}
}

\newglossaryentry{gap width distribution}
{
    name=gap width distribution,
    description={how gap widths relate to layer deposition times}
}

\newglossaryentry{resolution}
{
    name=resolution,
    description={the width of the gap containing a pruned layer or immediately following a retained layer (may be zero in the case of two successive retained strata)}
}

\newglossaryentry{stratum retention policy algorithm}
{
    name=stratum retention policy algorithm,
    description={the decision-making procedure of which strata to prune at each time point (also referred to simply as ``policy'')}
}

\newglossaryentry{policy resolution guarantee}
{
    name=policy resolution guarantee,
    description={upper bounds on resolutions across layers of a historical record with respect to layer age and/or record depth}
}

\newglossaryentry{extant record size}
{
    name=extant record size,
    description={the quantity of strata retained within a historical record at a particular time point}
}

\newglossaryentry{extant record order of growth}
{
    name=extant record order of growth,
    description={the asymptotic scaling relationship between the extant record size and record depth}
}

\newglossaryentry{pruning enumeration}
{
    name=pruning enumeration,
    description={calculation of the set of strata to be pruned at a particular time point under a retention policy}
}

\newglossaryentry{policy enactment}
{
    name=policy enactment,
    description={the act of performing pruning enumeration and deleting strata with enumerated deposition times}
}

\newglossaryentry{update}
{
    name=update,
    description={the process of performing a deposition and applying policy enactment}
}

\newglossaryentry{update time complexity}
{
    name=update time complexity,
    description={the scaling relationship associated with the number of computational operations necessary to perform an update}
}

\newglossaryentry{founding stratum}
{
    name=founding stratum,
    description={the first stratum deposited into the historical record, the oldest stratum}
}

\newglossaryentry{newest stratum}
{
    name=newest stratum,
    description={the most recent stratum deposited into the historical record}
}

\newglossaryentry{extant record}
{
    name=extant record,
    description={the set of strata that have been retained through the policy at the present time point}
}

\newglossaryentry{extant record enumeration}
{
    name=extant record enumeration,
    description={calculation of deposition times present in the extant record at a time point under a retention policy}
}

\newglossaryentry{policy self-consistency}
{
    name=policy self-consistency,
    description={the requirement for each deposition time within an extant record enumeration to be consistently present in all extant enumerations since that deposition time}
}

\newglossaryentry{historical record}
{
    name=historical record,
    description={refers to the set of layers defined up to the current time point, also referred to as a ``record''}
}

\newglossaryentry{hereditary stratigraphic column}
{
    name=hereditary stratigraphic column,
    description={container for a historical record --- in phylogenetic applications, associated with a digital population member, also referred to as a ``column''}
}

\newglossaryentry{annotation}
{
    name=annotation,
    description={the one-to-one association of hereditary stratigraphic records with individual digital organisms to facilitate phylogenetic analysis}
}

\newglossaryentry{inheritance}
{
    name=inheritance,
    description={the act of copying the parent organism's hereditary stratigraphic column annotation and performing an update to create the offspring organism's hereditary stratigraphic column annotation during a reproduction event}
}

\newglossaryentry{inference}
{
    name=inference,
    description={best-effort estimation of historical phylogenetic relationships from extant hereditary stratigraphic columns}
}

\newglossaryentry{perfect tracking}
{
    name=perfect tracking,
    description={maintenance of an exact record of phylogenetic events during an evolutionary simulation}
}

\newglossaryentry{streaming disposal problem}
{
    name=streaming disposal problem,
    description={poses the question of how to satisfactorily maintain a temporally representative collection of stored observations on a rolling basis as new observations stream in}
}

\begin{document}

\title{ Analysis of Phylogeny Tracking Algorithms for Serial and Multiprocess Applications }
\author{
    Matthew Andres Moreno\textsuperscript{1,2,3} \thanks{Corresponding author: \texttt{morenoma@umich.edu}} \quad
    Santiago Rodriguez Papa\textsuperscript{4} \quad
    Emily Dolson\textsuperscript{4,5} \quad
}
\date{}

\newcommand{\affil}[1]{\textsuperscript{#1}}
\newcommand{\affiliations}{
\affil{1} Ecology and Evolutionary Biology, University of Michigan, Ann Arbor, United States \\
\affil{2} Center for the Study of Complex Systems, University of Michigan, Ann Arbor, United States \\
\affil{3} Michigan Institute for Data Science, University of Michigan, Ann Arbor, United States
\affil{4} Department of Computer Science and Engineering, Michigan State University, East Lansing, United States \\
\affil{5} Ecology, Evolution, and Behavior, Michigan State University, East Lansing, United States \\
}

\maketitle

\begin{center}
\affiliations
\end{center}

\begin{abstract}


Since the advent of modern bioinformatics, the challenging, multifaceted problem of reconstructing phylogenetic history from biological sequences has hatched perennial statistical and algorithmic innovation.
Studies of the phylogenetic dynamics of digital, agent-based evolutionary models motivate a peculiar converse question: how to best engineer tracking to facilitate fast, accurate, and memory-efficient lineage reconstructions?
Here, we formally describe procedures for phylogenetic analysis in both serial and distributed computing scenarios.
With respect to the former, we demonstrate reference-counting-based pruning of extinct lineages.
For the latter, we introduce a trie-based phylogenetic reconstruction approach for ``hereditary stratigraphy'' genome annotations.
This process allows phylogenetic relationships between genomes to be inferred by comparing their similarities, akin to reconstruction of natural history from biological DNA sequences.
Phylogenetic analysis capabilities significantly advance distributed agent-based simulations as a tool for evolutionary research, and also benefit application-oriented evolutionary computing.
Such tracing could extend also to other digital artifacts that proliferate through replication, like digital media and computer viruses.

\end{abstract}

\clearpage
\newpage

\begin{bibunit}

\section{Introduction} \label{sec:introduction}

Biological phylogenetic history is staggeringly vast and deep.
Prokaryotes alone have a contemporary population size on the order of $10^{30}$ cells \citep{whitman1998prokaryotes}, and the phylogenetic record stretches back on the order of billions of years \citep{arndt2012processes}.
In addition to addressing questions of natural history, access to the phylogenetic record of biological life has proven informative to conservation biology, epidemiology, medicine, and biochemistry among other domains \citep{faithConservationEvaluationPhylogenetic1992, STAMATAKIS2005phylogenetics, frenchHostPhylogenyShapes2023,kim2006discovery}.

Although trifling in scale by comparison, computer simulations of evolution can generate vast histories in their own right, with common agent-based models achieving on the order of 200 million replication cycles per day \citep{ofria2009artificial}.
Distillation of this onslaught of history is necessary for data management tractability \citep{dolson2020interpreting}.
Nonetheless, existing analyses of phylogenetic structure within digital systems have already proven valuable, enabling diagnosis of underlying evolutionary dynamics \citep{moreno2023toward,hernandez2022can,shahbandegan2022untangling, lewinsohnStatedependentEvolutionaryModels2023a} and even serving as mechanism to guide evolution in application-oriented domains \cite{lalejini2024phylogeny,lalejini2024runtime,murphy2008simple,burke2003increased}.

Here, we formalize data structures and algorithmic procedures for space-efficient aggregation of phylogenetic history from evolutionary simulations on a rolling basis, and investigate their runtime performance characteristics.
In this paper, we focus in particular on asexual lineages, in which each entity has exactly one parent --- as opposed to sexual lineages, in which entities may have more than one parent.
Although they have historically underpinned a substantial proportion of evolutionary computation and simulation systems \citep{koza1994genetic,jefferson1990evolution} and continue to be of great interest within the community \citep{dang2018escaping}, they present unique challenges in phylogenetic tracking that merit methodological treatment in their own right \citep{godin2019apoget,moreno2024methods,mcphee2018detailed}.


We cover two approaches: (1) tracking lineages directly using a centralized data structure and (2) a decentralized, reconstruction-based approach called hereditary stratigraphy.
The former approach allows perfect accuracy, and we show that it is efficient when targeting single-processor applications (\textit{i.e.} using serial simulation).
The latter, in contrast, was designed for decentralized many-processor parallel/distributed simulation.
It supports user-controlled trade-offs between reconstruction precision and runtime resource overhead (e.g., memory, bandwidth).
We expect that hereditary stratigraphy indeed performs best in parallel/distributed contexts.

Accompanying public-facing open source Python packages provide convenient, plug-and-play access to phylogenetic tracking methodology across simulation systems \citep{moreno2022hstrat,dolson2023phylotrackpy}.
In this manner, methods described here promise to --- and, indeed, already have --- directly enable simulation-based evolution research.
On a purely algorithmic level, procedures and, in particular, representational considerations involved in direct phylogenetic tracking pertain also to more general issues of phylogenetic computation.
Likewise, data-management processes necessary to decentralized phylogenetic tracking relate directly to broader questions involving on-the-fly binning within the domain of stream processing \citep{moreno2024algorithms}.

Having motivated applications and algorithmic analyses of direct and decentralized tracking, we next present brief introductions to each methodology's operation.

\subsection{Direct Ancestry Tracking}

Most work on ancestry trees of self-replicating digital agents relies on centralized lineage tracking \citep{friggeri2014rumor,cohen1987computer,dolson2023phylotrackpy}.
Direct approaches to tracking replicator provenance in digital systems operate on this graph structure directly, distilling it from the full set of parent-child relationships over the history of a population to produce an exact historical account.

Without further regard, naive complete lineage tracking performs poorly for large-scale evolutionary systems.
For long-lived simulations, practical memory limitations preclude comprehensive records of replication events, which accumulate linearly with elapsed generations and population size.
To achieve a tractable space complexity, extinct lineages can be pruned.
Within serial processing contexts, an efficient reference-counting based approach may be applied.
We discuss this technique further in Section \ref{sec:perfect-tracking-algorithm}.

\subsection{Decentralized Ancestry Tracking}

Unfortunately, computational scale --- i.e., parallel/distributed computing --- erodes the simplicity, efficiency, and effectiveness of centralized tracking.
Detecting lineage extinctions requires either (1) collation of all replication and destruction events to a centralized data store or (2) peer-to-peer transmission of extinction notifications that unwind a lineage's (possibly many-hop) trajectory across host nodes.
Both options introduce runtime communication overhead.
To boot, the perfect-tracking paradigm is fragile to data loss.
Even a single failed communication event can entirely disjoin knowledge of how large portions of phylogenetic history relate.
As data loss is ubiquitous at scale \citep{cappello2014toward,ackley2011pursue}, this limitation is of serious concern.

These concerns motivated development \textit{hereditary stratigraphy}, of an alternate, fully-decentralized approach to phylogenetic tracking \citep{moreno2022hereditary}.
This methodology uses special genome annotations, termed \textit{hereditary stratigraphic columns}.
These annotations facilitate fast, accurate post hoc inference of phylogenetic relationships among evolved agents, akin to how genetic material enables phylogenetic inference in biology.

The core mechanism of hereditary stratigraphy is accretion, and subsequent inheritance, of a new randomized data packet onto column annotations each generation.
These stochastic fingerprints, which we call ``differentia,'' serve as a sort of checkpoint for lineage identity at a particular generation.
At future time points, extant annotations will share identical differentia for the generations during which they experienced shared ancestry.
Thus, the first mismatched fingerprint between two annotations bounds the recency of their most recent common ancestor (MRCA).

To circumvent annotation size bloat, hereditary stratigraphy prescribes a ``pruning'' process to delete differentia on the fly as generations elapse.
This pruning, however, reduces precision of resulting phylogenies.
The last generation of common ancestry between two lineages can be resolved no finer than retained checkpoint times.
In the context of hereditary stratigraphy, we refer to the procedure for down sampling as a ``stratum retention algorithm'' and the resulting patterns of retained differentia as a ``stratum retention policy.''
Stratum retention algorithms must decide how many records to discard, but also how remaining records should be distributed over past time.
Tuning of stratum retention trade-offs is discussed in \citet{moreno2024algorithms} as an instance of the more general ``stream curation'' problem.
Here, we build on these results to discuss the use of hereditary stratigraphy specifically in a phylogenetic context.

A typical objective for phylogeny-oriented application of hereditary stratigraphy is synthesis of population-level ancestry history.
Although following from similar principles as pairwise MRCA estimation, procedures to reconstruct the ancestry tree of many stratigraph annotations merit some further elaboration.
Section \ref{sec:reconstruction-algorithm} provides an agglomerative, trie-based algorithm for this task.

\subsection{Outline}

Remaining exposition in this paper is structured as follows:
\begin{itemize}
\item Section \ref{sec:reconstruction-algorithm} contributes a recently-developed algorithm for full-tree reconstruction from hereditary stratigraphic annotation data and analyzes its runtime characteristics.
\item Section \ref{sec:perfect-tracking-algorithm} formally presents an algorithm for perfect phylogenetic tracking with analysis of its time and space complexity.
\item Section \ref{sec:discussion} discusses which situations better suit perfect tracking vs. hereditary stratigraphy.
\item Section \ref{sec:conclusion} reflects on broader implications and future work.
\end{itemize}

\section{Phylogenetic Inference Algorithms} \label{sec:reconstruction-algorithm}

This section describes methodology to reconstruct patterns of descent among hereditary stratigraphic annotations.  
For an analysis of the time and space complexity of propagating the annotations at simulation runtime, see \citep{moreno2024algorithms}.
Sections \ref{sec:distance-based-reconstruction} and \ref{sec:trie-based-reconstruction} present a naive and a more apt approach, respectively, to inferring hierarchical relatedness (i.e., phylogeny) among hereditary stratigraphic annotations.
Phylogenetic reconstruction requires synthesis of relatedness relationships across an entire population to assemble a holistic historical account.  
To lay groundwork for this discussion, Section \ref{sec:pairwise-relatedness} first explores relatedness estimation between individual pairs of hereditary stratigraphic annotations.

\subsection{Pairwise Relatedness}
\label{sec:pairwise-relatedness}
Recall that hereditary stratigraphic annotations consist of a sequence of inherited ``checkpoint'' fingerprint differentia.
These sequences have a new randomly-generated differentia appended each generation.
Independent lineages accrue probabilistically-distinct differentia.
Relatedness estimation between annotations derives from a simple principle: mismatching differentia values at a time point indicates divergence of any two annotations' lineages.

Assuming both annotations employed the same streaming curation policy algorithm, if they share identical record depth they will have retained strata from identical time points.
What if one annotation has greater record depth (i.e., more generations elapsed) than the other?
We can truncate any strata beyond the depth of the younger annotation --- we already know no common ancestry will be shared at those time points.
Due to the self-consistency requirements of streaming curation, the younger annotation's curated time points will now be a superset of those of the older annotation.%
\footnote{During its excess stratum accumulation, the older annotation may have discarded some time points.}
So, we will search for the first mismatching stratum among the deeper annotation's early time points.

The possibility of spurious collisions between differentia (i.e., identical values by chance) complicates any application of binary search to identify the earliest time point with mismatched strata.
Consider lineage divergence as a boolean predicate: it evaluates false for all strata before some threshold of true lineage divergence and then true for all strata after.
Spurious collisions introduce the possibility of false negatives into search for this predicate's satisfaction threshold.
Take $c$ as retained stratum count per annotation.
If probabilistic confidence were acceptable, sufficient differentia could be tested at each binary search step to bound the probability of misidentifying the point of divergence.
The net failure rate depends on the number of opportunities for false negative divergence detections.
In the worst case, this number of opportunities will be $\mathcal{O}(\log c)$.
However, absolute certainty in determining the earliest differentia discrepancy will require worst case $\mathcal{O}(c)$ comparison (i.e., all differentia pairs when two annotations share no mismatching differentia).

Spurious collisions introduce false relatedness, causing a second complication: systematic overestimation of relatedness.
Expected bias can be readily calculated, as spurious collision probability stems from the number of unique differentia values.
As such, expected bias may be subtracted out to satisfy statistical analyses requiring mean-unbiased relatedness estimation.

\subsection{Distance-based Reconstruction}
\label{sec:distance-based-reconstruction}

The ease of calculating pairwise relatedness lends a straightforward option for whole-tree estimation: distance-based tree construction methods.
Such methods, like neighbor joining and UPGMA \citep{peng2007distance}, operate simply on pairwise distance estimates between taxa.
This distance-based approach was used in early work with hereditary stratigraphy \citep{moreno2022hereditary}, and is packaged within the acompanying \texttt{hstrat} library.

All-pairs comparisons necessary to prepare the distance matrix make such reconstructions at least $\mathcal{O}(c n^2)$, with $n$ as population size and $c$ as retained stratum count per annotation.
As will be shown presently, better results can be achieved by working directly with hereditary stratigraphic annotations' underlying structure.

\subsection{Trie-based Reconstruction}
\label{sec:trie-based-reconstruction}

The objective of phylogenetic reconstruction can be interpreted as production of a tree structure where leaf nodes share a common path from the root for the duration of their common ancestry.
Anatomically, hereditary stratigraphic annotations share common differentia up through the end of common ancestry.
Restated, annotations share a common prefix until the point of lineage divergence.
This observation suggests application of a trie data structure \citep{fredkin1960trie} to phylogenetic reconstruction.

As a preliminary simplification to be relaxed later, assume our population of $n$ annotations share consistent record depth.
Assuming identical retention policy algorithms, annotations will then also have consistent retained stratum count $c$.
Phylogenetic reconstruction through trie building follows as $\mathcal{O}(c n)$ \citep{mehta2018handbook}.
Such a trie building procedure corresponds each differentia within an annotation to a trie node.
The first annotation to be added unfolds its differentia into chronological linked list from most to least ancient.
Subsequent annotations agglomerate onto the trie by tracing down from the most-ancient root along the path its differentia sequence matches to, and forking off at the point the next-to-add differentia value is not available in the existing trie.
As an anecdotal reference for computational intensity, consider recent work using the Python-based \texttt{hstrat} library trie-building implementation \citep{moreno2023toward}.
This work achieved reconstructions over a population of 32,768 ($2^15$) synchronous $\approx 1,200$ stratum annotations within about five minutes wall time.

Two reconstruction biases should be noted.
First, because spurious differentia collisions introduce bias toward overestimation of relatedness, as noted above in Section \ref{sec:pairwise-relatedness}, branching events in the reconstructed tree will --- on average --- be more recent than in reality.
The expected rate of spurious collision is easily predictable, so this bias can be readily modeled and corrected for.
Second, trie reconstruction can overrepresent polytomies (i.e., internal multifurcations).
Branches that may have unfolded as separate events but fall within the same uncertainty gap within annotations' historical records will all lump together into a single polytomy.
This bias can be counteracted by splitting polytomies into arbitrary bifurcations with zero-length edges.

Allowing for uneven record depth among annotations complicates trie-based reconstruction.
As described in Section \ref{sec:pairwise-relatedness}, time points retained within deeper annotations effectively subset time points within younger annotations.
Thus, arranging youngest-first insertion order for trie construction ensures that as progressively deeper annotations are added, no inserted annotation retroactively injects a trie node between existing nodes.
In contrast, under deepest-first order, a younger inserted annotation may contain a time point not present among its lineage path in the trie (i.e., the older annotations having discarded it under streaming curation.
We therefore choose youngest-first insertion order to preclude scenarios where it would be necessary to splice a split away tendril back into the trie.

Under youngest-first order, an insertion may reach a trie position where the time point of the focal annotation's next-to-insert differentia skips beyond the time point of a next-to-consider child trie node.
The inserted annotation's true differentia value at that next-considered time point, having been discarded, is unknown.
This missing information is conceptually equivalent to a query string with a wildcard position \citep{fukuyama2016partial}.
Such wildcard queries can require evaluation of many branch paths during insertion, instead of just one.
An inserted taxon's lineage could proceed along any of the outgoing edges from the trie node preceding the wildcard.
Among possible paths, the path with the longest successive streak of strata consistent with the inserted annotation is best-evidenced.
Unfortunately, in the case of multiple wildcard positions, finding the best-evidenced path requires exploring a number of alternate paths potentially exponential with respect to maximum stratum count $c$ (i.e., maximum trie depth).
Taking the number of possible differentia values $d$ into account as the maximum branching factor, the worst case time complexity devolves to $\mathcal{O}(c d^c n)$ \citep{fukuyama2016partial}.
Calculation of the average case depends on the streaming curation policy algorithm at play and the underlying phylogenetic structure being reconstructed, which introduces analytical complexity and likely varies significantly between use cases.
Fortunately, the number of wildcard sites is limited due to: 1) record depth similarity, and 2) the tendency for there to be relatively few deep branches in most phylogenies, due to coalescence \citep{nordborgCoalescentTheory2019, berestyckiRecentProgressCoalescent2009}.
These factors likely reduce time costs.
Experimental performance evaluation with annotations derived from representative phylogenies will be warranted to explore the in-practice run time of wildcarding pruned-away strata.


\section{Perfect Tracking Algorithm} \label{sec:perfect-tracking-algorithm}

Traditionally, the problem of recording the phylogenies of replicating digital populations has been solved using a ``perfect tracking'' approach where all replication events are recorded and stored in a tree (or forest, if there are multiple roots) \citep{dolson2023phylotrackpy}.
In this section, we present results on the time and space complexity of perfect tracking;
these results will serve as a basis of comparison for hereditary stratigraphy's curation policy and reconstruction algorithms in Section \ref{sec:discussion}.

Perfect tracking algorithms are designed to augment an ongoing replication process by performing supplementary record-keeping.
As such, perfect tracking algorithms execute three separate routines: 1) initialization, 2) handling the birth of new individuals, and 3) (optionally) handling the removal of existing individuals.
The initialization step happens once before the founding set of population members is introduced.
The birth handling routine executes any time an individual is replicated.
Similarly, the removal-handling routine gets called any time an individual is removed from the set of currently active individuals (i.e., the population).

Perfect tracking can be used regardless of whether each individual has one or multiple parents (i.e., asexual vs. sexual reproduction).
Our analysis here will focus on the asexual scenario where each individual has only parent, but most of our results generalize to the multi-parent case.
The primary exception is the results presented in Section \ref{sec:perfect-tracking-pruning-space}, which will underestimate memory use in the average multi-parent case.

Perfect tracking can be done at an arbitrary level of abstraction \citep{dolson2020interpreting}.
For example, individual population members could be grouped into more abstract taxonomic units on the basis of a particular genetic or phenotypic trait.
As in biological phylogenetics, ancestry trees can show relationships between individuals, species, or any other taxonomic unit.
Increasing the level of abstraction decreases memory use by a constant (but often large) factor.
Here, we will assume that the entities being tracked are the underlying replicating individuals, as this case is the most computationally expensive (i.e., worst case).

In the following subsections, we present two different algorithms for perfect phylogenetic tracking: one which naively accumulates records for every individual that ever existed and another which reclaims memory as lineages go extinct.

\subsection{Naive perfect tracking}
\label{sec:naive-perfect-tracking}

The naive perfect tracking algorithm is fairly straightforward.
It is formally described in Algorithm \ref{alg:perfect-tracking}.

\renewcommand{\algorithmiccomment}[1]{ \newline \textbf{comment:} \textit{#1}}
\renewcommand{\algorithmicrequire}{\textbf{Input:}}

\begin{algorithm}
    \caption{Perfect phylogenetic tracking}
    \label{alg:perfect-tracking}
    \textbf{Part 1: Initialization}
    \begin{algorithmic}[1]
    \State{Define $T:=$ an empty forest}
    \end{algorithmic}
    \textbf{Part 2: Handle birth event}
    \begin{algorithmic}[1]
        \Require{$P_i$ = newly created organism, $P_j$ = $P_i$'s parent}
        \State{Add node representing $P_i$ to $T$}
        \State{Add edge in $T$ from $P_j$ to $P_i$}
    \end{algorithmic}
\end{algorithm}

\subsubsection{Time complexity}

The naive perfect tracking algorithm can be implemented in constant time (see Theorem \ref{thm:perfect-tracking-time}).

\begin{theorem}{Naive Perfect Tracking Time Complexity}
\label{thm:perfect-tracking-time}
The naive perfect tracking algorithm can be implemented in constant time ($\mathcal{O}(1)$) per birth event.
\end{theorem}

\begin{proof}
\label{prf:perfect-tracking-time}
For the purposes of this proof, we use the RAM model of computation. 
Initialization takes a single simple operation and only happens once. 
Thus, it is trivially constant-time.
Handling a birth event takes two simple operations.
Thus, each run of the birth-handling algorithim trivially takes constant time.
The birth-handling algorithm will run once per birth event.
\end{proof}

Note that, although the number of birth events is likely not constant with respect to the size of the population or the number of generations,
these birth events are not part of the tracking algorithm and would have happened regardless of whether phylogeny tracking was in place.
Thus, we need only consider time complexity per birth event.

\subsubsection{Performance in parallel and distributed environments}
\label{sec:perfect-tracking-distrbuted}

Naive perfect tracking can be adapted to parallel computing environments using universally unique agent identifiers through either of two straightforward options.
In the first approach, records of all agent creation and elimination events are collected to a centralized database as they occur.
In the second, each processor maintains an independent data file of local agent creation/elimination that are stitched together in postprocessing.
For the former, race conditions might occur where the phylogeny tracker is notified of a birth event involving a parent it does not yet know about.
This situation can be mitigated fairly cheaply by having a pool of ``waiting'' nodes that will be added once the appropriate parent node is added.

For both approaches to naive perfect tracking, bigger problems occur in distributed environments, where passing messages to a centralized data structure could become expensive \citep{moreno2022hereditary} or limited storage space is available to individual processor elements.
Worse still for either approach to naive perfect tracking is the possibility of data loss, which can become a significant likelihood at scale \citep{cappello2014toward}.
It is unclear how a perfect tracking algorithm could recover from missing data.

\subsubsection{Space complexity}

Naive perfect tracking has a gargantuan space complexity, owing to the fact that it maintains a tree containing a node for every object that ever existed (see Theorem \ref{thm:perfect-tracking-space}).

\begin{theorem}{Naive Perfect Tracking Space Complexity}
\label{thm:perfect-tracking-space}
The space complexity of the naive perfect tracking algorithm is $\mathcal{\theta}(\sum_{t=0}^{G} N_t$), where $N_t$ is the number of new objects born at time $t$ and $G$ is the number of time steps.
\end{theorem}

\begin{proof}
\label{prf:perfect-tracking-space}
The birth-handling algorithm adds a node to the tree ($T$) each time it is called.
It is called once for every birth event.
No nodes are ever removed from $T$.
Therefore, the number of nodes in $T$ is the total number of nodes ever created.
\end{proof}

\subsection{Pruning-enabled perfect tracking}
\label{sec:naive-perfect-tracking-with-pruning}

For most substantive applications, the amount of memory used by naive perfect tracking is prohibitive.
As such, it is common to use a ``pruning''-enabled algorithm that deletes records associated with extinct lineages (see Algorithm \ref{alg:perfect-tracking-pruning}). 
This procedure removes all parts of the phylogeny that are no longer ancestors of currently extant objects.
Thus, like naive perfect tracking, it introduces no uncertainty into down-stream analyses.

\begin{algorithm}
    \caption{Perfect phylogenetic tracking with pruning}
    \label{alg:perfect-tracking-pruning}
    \textbf{Part 1: Initialization} \\
    Same as algorithm \ref{alg:perfect-tracking} \\
    \textbf{Part 2: Handle birth event}
    \begin{algorithmic}[1]
        \Require{$P_i$ = newly created organism, $P_j$ = $P_i$'s parent}
        \State{Add node representing $P_i$ to $T$}
        \State{Add edge in $T$ from $P_i$ to $P_j$}
        \State{Record that $P_i$ is alive}
        \State{Increment $P_j$'s count of living offspring lineages}
    \end{algorithmic}
    \textbf{Part 3: Handle removal event}
    \begin{algorithmic}[1]
        \Require{$P_i$ = newly removed organism}
        \State{$P_j$ $\gets$ parent of $P_i$}
        \State{Record that $P_i$ is not alive}

        \While{$P_i$ is not alive \textbf{and} $P_i$'s count of living offspring lineages == 0}
            \State{Remove $P_i$ from $T$}
            \State{Decrement $P_j$'s count of living offspring lineages}
            \State{$P_i$ $\gets$ $P_j$}
            \State{$P_j$ $\gets$ parent of $P_j$}
        \EndWhile

    \end{algorithmic}
\end{algorithm}

\subsubsection{Time Complexity}

The pruning-enabled algorithm has amortized constant time complexity (see Theorem \ref{thm:perfect-tracking-with-pruning-time}).

\begin{theorem}{Pruning Time Complexity}
\label{thm:perfect-tracking-with-pruning-time}
The time complexity of pruning in perfect phylogenetic tracking is $\mathcal{O}(1)$, amortized. 
\end{theorem}

\begin{proof}
\label{prf:perfect-tracking-with-pruning-time}
The initialization and birth-handling steps of the pruning algorithm are trivially constant-time, as explained in proof \ref{prf:perfect-tracking-time} (technically the pruning algorithm adds two constant-time operations to the birth-handling step, but this fact is inconsequential).

Next, we prove that the removal-handling pruning algorithm runs in amortized constant time.
When an arbitrary taxon $t$ is removed, there are three possibilities: 1) $t$ has living descendants, 2) $t$ has no living descendants, but $t$'s parent is either alive or has other descendants that are alive, or 3) $t$ was the last living descendant of its parent.
In case 1, nothing can be removed from the phylogeny.
Case 1 takes contant time, because it only requires executing a single comparison operation.
In case 2, only $t$ can be removed from the phylogeny.
Case 2 also takes constant time, as it only requires removing a single node from the tree and doing three comparison operations.

In case 3, pruning must be done.
We recurse back up the lineage until we find a taxon, $a$, that is either alive or has living descendants. 
Let the distance between $t$ and the root of the tree be $d_t$ and the distance between $a$ and the root of the tree be $d_a$.
The pruning operation, then, takes $\mathcal{O}(d_t - d_a)$ steps. 
In the worst case, this value will be equal to the number of elapsed generations.

However, for case 3 to occur, all taxa from $t$ to $a$ (including $a$ but not including $t$) must have already been removed.
Consequently, $d_t - d_a$ case 1 removal operations must happen for every case 3 removal operation.
Thus, the amortized time to remove the sequence of taxa from $t$ to $a$ (including both $t$ and $a$) is:

\[
\frac{2(d_t - d_a) + 1}{d_t - d_a + 1}
\]

Since

\[
\lim_{d_t - d_a\to\infty} \frac{2(d_t - d_a) + 1}{d_t - d_a + 1} = 2
\],

the amortized time complexity of pruning is bounded by a constant, i.e., it is $\mathcal{O}(1)$.

\end{proof}

\subsubsection{Performance in parallel and distrbuted environments}
\label{sec:perfect-tracking-pruning-distrbuted}

Like naive perfect tracking, pruning-enabled perfect tracking can proceed via either a fully-centralized database or through localized record-keeping.

Centralized pruning-enabled perfect tracking has all the same issues in parallel and distributed environments that naive perfect tracking has.
It also adds an additional cost: parallel processes must be synced up before a pruning operation can be performed.
Otherwise, the algorithm risks incorrectly pruning a lineage that it will later learn produced a new artifact.
This requirement would translate to a substantial speed cost in practice.

In a localized model, when the last remnant of an immigrant lineage is eliminated an extinction notification would need to be propagated back to the processing element that lineage immigrated from.
Records for lineages that extinct locally but had generated emigres to other processing elements would need to be retained until extinction notifications for all emigres are confirmed.
Even for lineages that span few processing elements, repeat migrations between a single pair of processors could wind up long breadcrumb trails that necessitate extensive cleanup operations.

\subsubsection{Space complexity}
\label{sec:perfect-tracking-pruning-space}

As with naive (i.e., unpruned) perfect tracking, the primary memory cost of pruned perfect tracking comes from maintaining the tree of records, $T$.
Thus, our analysis in this section will focus on the size of $T$.
Technically, the worst-case size order of growth of perfect tracking with pruning is the same as for naive perfect tracking (see Theorem \ref{thm:perfect-tracking-space}).
However, whereas that order of growth is also the best case for naive perfect tracking, it is a somewhat pathological case under pruning.
When using pruning, such a case would only occur when every object is copied before being removed (i.e., no lineages go extinct).
There are some legitimate reasons these cases might occur, but most of them are fairly esoteric (see Section \ref{sec:perfect-tracking-space-supp}).


With pruning, the asymptotic behavior of the average case turns out to be substantially better than that of the worst case.
Luckily, in the vast majority of applications, the average case is more practically informative than the asymptotic behavior of the worst case.
Here, we investigate the expected size of $T$ (which we will call $E(|T|)$) in the average case as the population size ($N$) and number of time points ($G$) increase.
Of course, investigating the average case requires making some assumptions about the scenarios where this algorithm may be used.
An expedient assumption to make is that objects to be replicated are selected at random from the population (i.e., drift).
This process is described by a neutral model.
While this assumption is obviously not completely true in most cases, it will turn out that $E(|T|)$ is actually smaller for most realistic cases (e.g. any evolutionary scenario with directional selection).

Conveniently, the asymptotic behavior of $E(|T|)$ as $N$ increases has been studied extensively by evolutionary biologists under a variety of random selection models \citep{berestyckiRecentProgressCoalescent2009, tellierCoalescenceMultipleBranching2014, nordborgCoalescentTheory2019}.
The relevant branch of mathematics is coalescence theory, which describes the way particles governed by random processes agglomerate over time.
There are a variety of models of random selection, all of which yield different behavior.

Neutral models that could plausibly describe a process of replicating digital artifacts experience ``coalescence'' events.
In these events, the population establishes a new, more-recent common ancestor as a result of old lineages dying out.
This observation implies that when we run our perfect tracking with pruning algorithm for sufficiently long periods of time, on average we should expect our tree $T$ to have a long unifurcating ``stem.''
The most recent node in this chain will be the most recent common ancestor (MRCA) of the whole extant population.
The descendants of the MRCA will form a sub-tree, which we will call $T_c$.

The expected size of $T_c$ (denoted $E(|T_c|)$ here), then, will be the dominant factor affecting $E(|T|)$.
$E(|T_c|)$ is closely related to a value commonly studied in coalescence theory literature called the expected total branch length (usually denoted $L_n$ or $T_{tot}$; we will refer to it as $L_n$).
$L_n$ assumes that the phylogeny is represented by a tree, $T'$, that contains the leaf nodes, MRCA, and any internal nodes from which multiple lineages diverge (i.e., no unifurcations).
Edges in this tree have weights that represent the amount of time between the origin of the parent and the origin of the child.
$L_n$ is the sum of these weights.
Although the asymptotic behavior of $L_n$ depends on specifics of the evolutionary process, it scales sub-linearly with respect to population size for all realistic models that have been investigated in the literature to our knowledge \citep{gnedinLcoalescentsSurvey2014}.%
\footnote{A possible exception is spatial models, which have received less formal analysis but are known to exhibit slower coalescence \citep{berestyckiRecentProgressCoalescent2009}}
How is this possible?
Under realistic neutral models, various subsets of the population will coalesce to common ancestors much faster than the whole population does \citep{nordborgCoalescentTheory2019}.
Consequently, there are many shallow branches.
The farther back in $T_c$ you go, the fewer branches there are.
As a result, the marginal impact of increasing $N$ diminishes as $N$ increases.

What do these observations tell us about $E(|T_c|)$?
Because $T_c$ contains unifurcations, each edge in $T'$ corresponds to a chain of unifurcating nodes in $T_c$.
The length of this chain in $T_c$ will, on average, be proportional to the weight of the corresponding edge in $T'$.
Thus, the expected value of $L_n$ is almost identical to $E(|T_c|)$.
However, there is one important difference.
The reason $T'$ is able to scale sub-linearly with $N$ is that as $N$ goes to infinity the lengths of the branches to the leaf nodes (i.e., the extant population) approach zero \citep{nordborgCoalescentTheory2019, delmasAsymptoticResultsLength2008, drmotaAsymptoticResultsConcerning2007}.
In contrast, the space taken up by each of the $N$ leaf nodes in $T_c$ is a fixed constant and cannot drop to zero.
Thus, $E(|T_c|)$ is $\mathcal{O}(L_n + N)$.
Note that the time to coalescence depends entirely on the population size ($N$); the total number of generations elasped since the start of the process ($G$) does not affect the asymptotic growth rate of $E(|T_c|)$.

The precise asymptotic scaling of $L_n$ with respect to population size $N$ depends on the exact model of random replication used \citep{tellierCoalescenceMultipleBranching2014}.
Under the best-studied model, Kingman's Coalescent, it is $\mathcal{O}(log(N))$ \citep{kingmanCoalescent1982, delmasAsymptoticResultsLength2008}.
Under another reasonable model, the Bolthausen-Sznitman coalescent, it is $\mathcal{O}(\frac{N}{log(N)})$ \citep{drmotaAsymptoticResultsConcerning2007}.
Because these growth rates are all less than $N$, $E(|T_c|)$ ends up being dominated by $N$ term.
Thus the overall memory usage of pruning-enable perfect tracking is $\mathcal{O}(N + G)$.
Trivially, the ``stem'' part of $T$ can be pruned off when coalescence occurs, producing a bound of $\mathcal{O}(N)$.

While $L_n$ is not the primary factor determining the memory use of this algorithm, it still has a substantial practical impact on resource usage.
Thus, it is worth briefly exploring our practical expectations of its size.
Earlier, we mentioned that in practice $T$ will usually be smaller than predicted by neutral models.
These models represent a scenario called ``neutral drift'' in evolutionary theory.
Under drift, coalescence can be slow because it needs to happen entirely by chance.
When some members of the population are more likely to be selected than others and that selective advantage is heritable (assuming population size $N$ is stable), we expect to see ``selective sweeps'' in which those members of the population outcompete others.
A selective sweep will usually speed up coalescence (unless it is competing with a different simultaneous selective sweep) \citep{neherGeneticDraftSelective2013}.
Thus, in general, in the presence of non-random selection, we expect $T_c$ to be much smaller than we would expect under drift.
Note, however, that the Bolthausen-Sznitman coalescent mentioned above does model scenarios with selection.

Would we ever expect $T_c$ to be larger than under drift?
Only if lineages coexist for longer than we would expect by chance.
Such a regime can theoretically occur when there is selection for stable coexistence or diversity maintenance.
However, such a regime also involves introducing non-random (balancing/stabilizing) selection.
Thus, in practice, while these regimes maintain multiple deep branches, there is often frequent coalescence within those branches.
Theoretically, though, extreme pressure for diversity could cause $E(|T|)$ to scale faster than $\mathcal{O}(N)$.
Such scenarios could potentially be created via ecological interactions.  

One factor that we have glossed over in this analysis is the possibility for $N$ to change over time.
For reasons that are explored further in Section \ref{sec:perfect-tracking-space-supp}, allowing population size $N$ to fluctuate complicates mathematical analysis but does not substantially alter our conclusions.

\section{Discussion} \label{sec:discussion}

In choosing between perfect tracking and hereditary stratigraphy, there are multiple factors to consider: 1) time complexity in the relevant model of computation, 2) memory use, and 3) amount of information retained.

\subsection{Time complexity}

In a non-distributed, serial computing environment, perfect tracking runs in constant time.
The time complexity of hereditary stratigraphy depends on the retention policy used, but scales at least linearly with annotation size.
So, if annotation size growth is allowed to maintain static reconstruction resolution guarantees in long-running processes, time complexity of perfect tracking will scale preferably to hereditary stratigraphy.
A lower bound on hereditary stratigraphy's time complexity comes from copying of annotation material from parent to child when new artifacts are born.
Consequently, the exact time complexity depends on the curation policy but can scale with number of generations $G$ elapsed, unless a policy with a constant-cap order of growth is used.
Note that hereditary stratigraphy incurs an additional one-time complexity cost due to the fact that its output generally requires post-processing before analysis (i.e., phylogenetic reconstruction).
Assuming comparable generational depth across lineages, reconstruction can be achieved with time complexity linear to annotation size and linear to tree tip count (which can optionally be held below population size through subsampling).

In contrast, in parallel or distributed environments, perfect tracking suffers from the problems discussed in Sections \ref{sec:perfect-tracking-distrbuted} and \ref{sec:perfect-tracking-pruning-distrbuted}.
Thus, hereditary stratigraphy usually runs faster in these environments and is generally the better choice.
Moreover, it is robust to data loss, making it reliable to use in situations where perfect tracking would break.

\subsection{Memory use}

As discussed in Section \ref{sec:perfect-tracking-pruning-space}, pruning-enabled perfect tracking is fairly memory efficient ($\mathcal{O}(N)$, where $N$ is population size) under a wide range of practical circumstances.
However, in practice, its memory use can vary dramatically over time.
This variation is partially due to the fact that the expected distribution of coalescence times is exponential, leading to substantial variation about the mean.
It is also due to the potential for large variation in tree size caused by the ecological and evolutionary dynamics at play in a given scenario being tracked.
Anecdotally, the memory cost of perfect tracking often poses a real-world obstacle to using it at full resolution.

Like perfect tracking, hereditary stratigraphy's memory footprint grows linearly $\mathcal{O}(N)$ with population size.
Under configurations with guarantee static MRCA resolution at arbitrarily generational depth, memory footprint also scales with $G$ --- either $\mathcal{O}(G)$ or $\mathcal{O}(\log G)$ depending on whether guarantees are absolute or recency-proportional \citep{moreno2024algorithms}.
However, hereditary stratigraphy can also be conducted $\mathcal{O}(1)$ constant with respect to $G$.
Unlike pruning-enabled perfect tracking, hereditary stratigraphy's memory usage is a hard guarantee.
Consequently, it can be a good choice in cases where large fluctuations in memory use are not acceptable.
It can also be helpful in circumstances that inhibit coalescence and start pushing perfect tracking towards worst-case performance.

\subsection{Accuracy}

As its name implies, perfect tracking returns completely accurate phylogenetic data.
Nevertheless, memory use can be tuned by increasing the granularity of data being recorded (i.e., make the taxonomic unit represented by each node more abstract).

Hereditary stratigraphy introduces imprecision that can potentially lead to inaccuracy.
At the same time, it introduces tools for precisely controlling the level of inaccuracy and its trade-off with memory and time costs.
As such, it can provide more granular scaleback of phylogenetic resolution than is possible through coarsening of taxonomic units under perfect tracking.
Hereditary stratigraphy is always tracked at the level of the individual (i.e., maximum taxonomic resolution), meaning it can capture phylogenetic events that perfect tracking would have to lump together due to memory constraints.


\subsubsection{Overall take-away}

Perfect tracking and hereditary stratigraphy are complimentary techniques that occupy different methodological niches.
In many simple cases, perfect tracking requires fewer resources and provides more accurate data.
However, hereditary stratigraphy can operate in computational environments that perfect tracking cannot.
It also enables new trade-offs to be made between data resolution, memory use, and data accuracy.
Lastly, while perfect tracking is often more memory efficient, hereditary stratigraphy can be configured to provide strong $\mathcal{O}(1)$ guarantees about memory use.



\section{Conclusion} \label{sec:conclusion}



We have documented the algorithmic backing of both direct phylogenetic tracking and the emerging ``hereditary stratigraphy'' methodology for distributed lineage tracking of replicating digital artifacts.

With respect to hereditary stratigraphy, we have here focused on the postprocessing algorithms used to compare and interpret hereditary stratigraphic annotations.
In particular, we present a trie-based tree reconstruction algorithm that proceeds from the underlying structure of the hereditary stratigraphy annotation.
As a point of comparison, we provide an algorithmic formalization of the established perfect-tracking algorithms currently standard in agent-based evolutionary modeling.
We find that both perfect tracking and hereditary stratigraphy have circumstances in which they excel.


Ultimately, our interest in extraction of phylogenetic history is application-driven.
To these ends, reference implementations of most algorithms described are provided in the public-facing \texttt{hstrat} and \texttt{phylotrackpy} Python libraries \citep{moreno2022hstrat, dolson2023phylotrackpy}.
These libraries enable application of perfect phylogenetic tracking and hereditary stratigraphy techniques in a few lines of code.

\putbib

\end{bibunit}

\clearpage
\newpage

\setcounter{section}{0}
\makeatletter
\renewcommand \thesection{S\@arabic\c@section}
\renewcommand\thetable{S\@arabic\c@table}
\renewcommand \thefigure{S\@arabic\c@figure}
\makeatother

\begin{bibunit}

\section{Appendix}


\subsection{Additional notes on the memory requirements of perfect tracking} \label{sec:perfect-tracking-space-supp}

Note that perfect phylogeny tracking algorithms are designed to plug into a larger computational process.
For phylogeny tracking to be relevant, this computational process must have a collection of $N$ objects that are currently eligilble to be copied (in evolutionary terms, the population).
Therefore, the space complexity of this process must be at least $\mathcal{O}(N)$.
As long as the expected size of $T$ scales no faster than $\mathcal{O}(N)$, then, phylogeny tracking will not be the primary factor determining memory usage.
While the memory cost of phylogeny tracking could still be significant in these cases, it is unlikely to be the primary factor determining whether running the program is tractable.

Previously, we noted that cases where perfect tracking with pruning achieves its worst case size order of growth are esoteric.
There are two ways that theses situations can occur:

\begin{itemize}
\item $N$ is continuously increasing.
\item There is effectively no coalescence.
\end{itemize}

In case 1, the space complexity of the external computational process will also be growing at least as fast as $\mathcal{O}(N)$.
Because $N$ itself is growing, this behavior should overwhelm the size of the history in $T$ in most cases.
Even in the (completely unrealistic) case where $N$ is allowed to reach infinity, there is an argument to be made that the internal nodes of $T$ would not be substantially contributing to memory complexity.
While this argument is beyond the scope of this paper, it would center around the fact that most evolutionarily-realistic coalescence proceses ``come down from infinity'' \citep{berestyckiRecentProgressCoalescent2009}.

Case 2 is a legitimate circumstance where worst case memory complexity could occur. 
It would occur in cases where every member of the population has (approximately) one offspring.
While this level of regularity is unlikely in the context of evolution, it could happen when tracking non-evolutionry digital artifacts.



\putbib

\end{bibunit}

\end{document}